\documentclass[11pt]{article}
\usepackage[english]{babel}
\usepackage{verbatim}
\usepackage{authblk}
\usepackage{subcaption}
\usepackage{amsmath, amsthm, amssymb, graphics, graphicx}
\DeclareMathOperator{\rwd}{rwd}
\DeclareMathOperator{\cwd}{cwd}
\DeclareMathOperator{\rk}{rk}
\DeclareMathOperator{\cutrk}{cutrk}

\newcommand{\mathsym}[1]{{}}
\newtheorem{theorem}{Theorem}[section]
\newtheorem{lemma}{Lemma}[section]

\newtheorem{proposition}[lemma]{Proposition}
\newtheorem{corollary}{Corollary}[section]

\newtheorem{observation}{Observation}

\title{\bf\Large{{Coloring even-hole-free graphs with no star cutset}}}
\author{Ngoc Khang Le 
\thanks{This work was performed within the framework of the LABEX MILYON (ANR-10-LABX-0070) of Universit\'e de Lyon, within the program ``Investissements d'Avenir'' (ANR-11- IDEX-0007) operated by the French National Research Agency (ANR). Partially supported by ANR project Stint under reference ANR-13-BS02-0007.}
\thanks{Email: ngoc-khang.le@ens-lyon.fr}
}
\affil{LIP, ENS de Lyon, Lyon, France}

\begin{document}
\maketitle

\begin{abstract}
A \emph{hole} is a chordless cycle of length at least $4$. A graph is \emph{even-hole-free} if it does not contain any hole of even length as an induced subgraph. In this paper, we study the class of even-hole-free graphs with no star cutset. We give the optimal upper bound for its chromatic number in terms of clique number and a polynomial-time algorithm to color any graph in this class. The latter is, in fact, a direct consequence of our proof that this class has bounded rank-width. 
\end{abstract}

\section{Introduction} \label{S:1}

All graphs in this paper are finite, simple and undirected. Let $F$ be a graph, we say that a graph $G$ is \emph{$F$-free} if it does not contain $F$ as an induced subgraph. Let $\mathcal{F}$ be a (possibly infinite) family of graphs. A graph $G$ is \emph{$\mathcal{F}$-free} if it is $F$-free, for every $F\in\mathcal{F}$. A \emph{hole} is a chordless cycle of length at least $4$. A hole is \emph{even} (\emph{odd}) if it contains an even (resp. odd) number of nodes. 

Let us first introduce perfect graphs, a graph class which has a very close relation with even-hole-free graphs and is the motivation to study this class. A graph $G$ is \emph{perfect} if for every induced subgraph $H$ of $G$, $\chi(H)=\omega(H)$, where $\chi(H)$ denote the \emph{chromatic number} of $H$, i.e. the minimum number of colors needed to color the vertices of $H$ so that no two adjacent vertices receive the same color, and $\omega(H)$ denotes the size of a largest clique in $H$, where a \emph{clique} is a graph in which all the vertices are pairwise adjacent. The famous Strong Perfect Graph Theorem (conjectured by Berge \cite{B61}, and proved by Chudnovsky, Robertson, Seymour and Thomas \cite{CRST2006}) states that a graph is perfect if and only if it does not contain an odd hole nor an odd antihole as an induced subgraph (where an \emph{antihole} is a complement of a hole). The graphs that do not contain an odd hole nor an odd antihole as an induced subgraph are known as \emph{Berge} graphs.

The structure of even-hole-free graphs was first studied by Conforti, Cornu\'ejols, Kapoor and Vu\v{s}kovi\'c in \cite{CCKV2002} and \cite{CCKV20022}. They were focused on showing that even-hole-free graphs can be recognized in polynomial time, and their primary motivation was to develop techniques which can then be used in the study of perfect graphs. In \cite{CCKV2002}, they obtained a decomposition theorem for even-hole-free graphs that uses 2-joins and star, double star and triple star cutsets, and in \cite{CCKV20022}, they used it to obtain a polynomial time recognition algorithm for even-hole-free graphs. That decomposition technique is actually useful since the Strong Perfect Graph Conjecture was proved in \cite{CRST2006} by decomposing Berge graphs using skew cutsets, 2-joins and their complements. Soon after,
the recognition of Berge graphs was shown to be polynomial by Chudnovsky, Cornu\'ejols, Liu, Seymour and Vu\v{s}kovi\'c in \cite{CCLSV2005}. A better decomposition theorem for even-hole-free graphs using only 2-joins and star cutsets was given in \cite{DV13}.

Finding a maximum clique, a maximum independent set and an optimal coloring are all known to be polynomial for perfect graphs \cite{GLS81,GLS2012}. However, these algorithms rely on the ellipsoid method, which is impractical. It is still an open question to find a combinatorial algorithm for these problems. On the other hand, the complexities of finding a maximum stable set and an optimal coloring are both open for even-hole-free graphs. Note that a maximum clique of an even-hole-free graphs can be found in polynomial time, since a graph without a hole of length $4$ has polynomial number of maximal cliques and one can list them all in polynomial time \cite{F89}. 

Therefore, we would like to see if the decomposition theorem can be used to design polynomial-time algorithms for all these combinatorial problems. The general answer should be impossible since there are some kinds of decomposition which do not seem to be friendly with these problems like star or skew cutsets. On the other hand, 2-joins look very promising. Indeed, in \cite{TV12}, Trotignon and Vu{\v{s}}kovi{\'c} already gave the polynomial algorithms to find a maximum clique and maximum independent set in the subclasses of even-hole-free and Berge graphs which are fully decomposable by only 2-joins (namely, even-hole-free graphs with no star cutset and perfect graphs with no balanced skew-partition, homogenous pair nor complement 2-join). In \cite{CTTV2015}, they generalize the result for Berge graphs to perfect graph with no balanced skew-partitions. Note that an $O(n^k)$ algorithm that computes a maximum weighted independent set for a class of perfect graphs closed under complementation, yields also an $O(n^{k+2})$ algorithm that computes an optimal coloring for the same class (see for instance \cite{KS97,S2003}). Hence, all three problems (clique, independent set and coloring) are solved for perfect graph with no balanced skew-partitions. However, the coloring problem for even-hole-free graphs with no star cutset remains open despite of its nice structure. In this paper, we prove that this class has bounded rank-width, a graph parameter which will be defined in the next section. This implies that it also has bounded clique-width (a parameter which is equivalent to rank-width in the sense that one is bounded if and only if the other is also bounded). Therefore, coloring is polynomial-time solvable for even-hole-free graphs with no star cutset by combining the two results: Kobler and Rotics \cite{KR2003} showed that for any constant $q$, coloring is polynomial-time solvable if a $q$-expression is given, and Oum \cite{O2008} showed that a $(8^p-1)$-expression for any $n$-vertex graph with clique-width at most $p$ can be found in $O(n^3)$. Note that our result is strong in the sense that it implies that every graph problem expressible in monadic second-order logic formula is solvable in polynomial-time for even-hole-free graphs with no star cutset (including also finding a maximum clique and a maximum independent set).   

We also know that even-hole-free graphs are $\chi$-bounded by the concept introduced by Gy\'arf\'as \cite{G87}: A class of graphs $\mathcal{G}$ is \emph{$\chi$-bounded} with \emph{$\chi$-bounding function $f$} if for every graph $G\in \mathcal{G}$, $\chi(G)\leq f(\omega(G))$. In \cite{ACHRS2008}, it is proved that $\chi(G)\leq 2\omega(G)-1$ for every even-hole-free graph $G$. One might be interested in knowing whether this bound could be improved for the class that we are considering, even-hole-free graphs with no star cutset. Let \emph{$\rwd(G)$} denote the rank-width of some graph $G$. The main results of our paper are the two following theorems: 

\begin{theorem} \label{T1}
Let $G$ be a connected even-hole-free graph with no star cutset. Then $\chi(G)\leq \omega(G)+1$.
\end{theorem}

\begin{theorem} \label{T2}
Let $G$ be a connected even-hole-free graph with no star cutset. Then $\rwd(G)\leq 3$.
\end{theorem}

The rest of our paper is organized as follows. In Section \ref{S:2}, we formally define every notion and mention all the results that we use in this paper. The proof of Theorem \ref{T1} is presented in Section \ref{S:3} and the proof of Theorem \ref{T2} is given in Section \ref{S:4}.

\section{Preliminaries} \label{S:2}

Let $G(V,E)$ be a graph. For $X\subseteq V(G)$, we denote by $G\setminus X$ the graph obtained from $G$ by removing all the vertices in $X$. In case $X=\{v\}$, we write $G\setminus v$ instead of $G\setminus \{v\}$. We also denote by $G[X]$ the subgraph of $G$ induced by some $X\subseteq V(G)$. For $v\in V(G)$, let $N_G(v)$ denote the set of neighbors of $v$ in $G$.  For $X\subseteq V(G)$, let $N_G(X)$ denote the set of vertices in $V(G)\setminus X$ adjacent to a vertex in $X$. We also write $N(v)$ or $N(X)$ instead of $N_G(v)$ or $N_G(X)$ if there is no ambiguity. Let $A\subseteq V(G)$ and $b\in V(G)\setminus A$, we say that $b$ is \textit{complete} to $A$ if $b$ is adjacent to every vertex in $A$. A \textit{clique} in $G$ is a set of pairwise adjacent vertices. A \textit{stable set}, or an \textit{independent set} in $G$ is a set of pairwise non-adjacent vertices.  A \textit{path} $P$ is a graph with vertex-set $\{p_1,\ldots,p_k\}$ such that either $k=1$, or for $i, j\in\{1,\ldots,k\}$, $p_i$ is adjacent to $p_j$ iff $|i-j| = 1$.  We call $p_1$ and $p_k$ the \textit{ends} of the path, $\{p_2,\dots,p_{k-1}\}$ its \textit{interior} and also call each vertex in $\{p_2,\dots,p_{k-1}\}$ \emph{interior vertex}. Let $P^*$ denote the path obtained from $P$ by removing its two ends. A \textit{flat} path in $G$ is a path such that all of the interior vertices are of degree $2$. A \textit{hole} $H$ is a graph with vertex-set $\{h_1,\ldots, h_k\}$ such that $k\geq 4$ and for $i,j\in\{1,\ldots,k\}$, $h_i$ is adjacent to $h_j$ iff $|i-j|=1$ or $|i-j|=k-1$. The \textit{length} of a path or a hole is the number of its edges. A hole of length $k$ is called a \emph{$k$-hole} Note that a path may have length $0$. A graph is \textit{even-hole-free} if it does not contain any hole of even length as an induced subgraph. Our proof heavily relies on the decomposition lemmas for even-hole-free graphs with no star cutset given by Trotignon and Vu{\v{s}}kovi{\'c} in \cite{TV12}. Hence, in the next part of this section, the formal definitions needed to state these lemmas will be given.

In a connected graph $G$, a subset of nodes is a \emph{cutset} if its removal yields a disconnected $G$.  A cutset $S\subseteq V(G)$ is a \emph{star cutset} if $S$ contains a node $x$ adjacent to every node in $S\setminus x$. A cutset $S\subseteq V(G)$ is a \emph{clique cutset} if $S$ is a clique. It is clear that clique cutset is a particular star cutset. The only vertex of a clique cutset of size $1$ is called the \emph{cut-vertex}.

A \emph{$2$-join} in a graph $G$ is a partition $(X_1, X_2)$ of $V(G)$ with specified sets $(A_1, A_2, B_1, B_2)$ such that the followings hold:
\begin{itemize}
\item $|X_1|, |X_2| \geq 3$.
\item For $i=1,2$, $A_i\cup B_i \subseteq X_i$, and $A_i$ and $B_i$ are nonempty and disjoint.
\item Every node of $A_1$ is adjacent to every node of $A_2$, every node of $B_1$ is adjacent to every node of $B_2$, and these are the only adjacencies between $X_1$ and $X_2$.
\item For $i=1,2$, the graph induced by $X_i$, $G[X_i]$, contains a path with one end in $A_i$ and the other in $B_i$. Furthermore, $G[X_i]$ does not induce a path.
\end{itemize}

In this case, we call $(X_1, X_2, A_1, B_1, A_2, B_2)$ a \textit{split} of $(X_1, X_2)$. We also denote by $C_i$ the set $X_i\setminus(A_i\cup B_i)$ for $i=1,2$. Since the goal of decomposition theorems is to break our graphs into smaller pieces that we can handle inductively, we need a way to construct them. \emph{Blocks of decomposition} with respect to a $2$-join (will be defined below) are built by replacing each side of the $2$-join by a path of length at least $3$ and the next lemma shows that for even-hole-free graphs, there exists a unique way to choose the
parity of that path.

\begin{lemma}[\cite{TV12}] \label{parity} 
Let $G$ be an even-hole-free graph and $(X_1, X_2, A_1, B_1, A_2, B_2)$ be a split of a $2$-join of $G$. Then for $i = 1, 2$, all the paths with an end in $A_i$, an end in $B_i$ and interior in $C_i$ have the same parity.
\end{lemma}

Let $G$ be an even-hole-free graph and $(X_1, X_2, A_1, B_1, A_2, B_2)$ be a split of a $2$-join of $G$. The \textit{blocks of decomposition} of $G$ with respect to ($X_1$, $X_2$) are the two graphs $G_1$, $G_2$ built as follows. We obtain $G_1$ by replacing
$X_2$ by a \textit{marker path} $P_2$ of length $k_2$, from a vertex $a_2$ complete to $A_1$, to a vertex $b_2$ complete to $B_1$ (the interior of $P_2$ has no neighbor in $X_1$). We choose $k_2=3$ if the length of all the paths with an end in $A_2$, an end in $B_2$ and interior in $C_2$ is odd (they have the same parity due to Lemma
\ref{parity}), and $k_2=4$ otherwise. The block $G_2$ is obtained similarly by replacing $X_1$ by a marker path $P_1$ of length $k_1$ with two ends $a_1$, $b_1$.

Now we present some definitions for the basic classes in the decomposition theorem for even-hole-free graphs. Let $x_1,x_2,x_3,y$ be four distinct nodes such that $x_1,x_2,x_3$ induce a triangle. A \emph{pyramid} is a graph induced by three paths $P_{x_1y}=x_1\ldots y$, $P_{x_2y}=x_2\ldots y$, $P_{x_3y}=x_3\ldots y$ such that any two of them induce a hole. By the definition, at most one of these paths is of length $1$. A pyramid is \textit{long} if all three paths are of length greater than $1$. Note that in an even-hole-free graph, the lengths of all these three paths have the same parity. 

An \emph{extended nontrivial basic} graph $R$ is defined as follows:
\begin{enumerate}
\item $V(R)=V(L)\cup\{x,y\}$.
\item $L$ is the line graph of a tree $T$.
\item $x$ and $y$ are adjacent, $x,y\notin V(L)$.
\item Every maximal clique of size at least $3$ in $L$ is called an \emph{extended} clique. $L$ contains at least two extended cliques.
\item The nodes of $L$ corresponding to the edges incident with vertices of degree one in $T$ are called \emph{leaf nodes}. Each leaf node of $L$ is adjacent to exactly one of $\{x,y\}$, and no other node of $L$ is adjacent to $\{x,y\}$.
\item These are the only edges in $R$.
\end{enumerate}

Note that the definition of the extended nontrivial basic graph we give here is simplified compared to the one from the original paper \cite{DV13} (since they  prove a decomposition theorem for a more general class, namely, $4$-hole-free odd-signable graphs), but it is all we need in our proof. The following property of $R$ is easy to observe in even-hole-free graphs with no star cutset:
\begin{lemma} \label{pr4} $x$ $($and $y)$ has at most one neighbor in every extended clique. Furthermore, if $x$ has some neighbor in an extended clique $K$, then $N(y)\cap K=\emptyset$.
\end{lemma}

\begin{proof}	
If $x$ has two neighbors $a$, $b$ in some extended clique $K$, then $N(a)\setminus\{b\}=N(b)\setminus\{a\}$, implying that there is a star cutset $S=(\{a\}\cup N(a))\setminus \{b\}$ in $R$ separating $b$ from the rest of the graph, a contradiction. Also, if $x$ and $y$ both have a neighbor in a same extended clique, called $a$ and $b$, respectively, then $\{x,a,b,y\}$ induces a $4$-hole, a contradiction.
\end{proof}

An even-hole-free graph is \textit{basic} if it is one of the
following graphs:
\begin{itemize}
	\item a clique,
	\item a hole,
	\item a long pyramid, or
	\item an extended nontrivial basic graph.
\end{itemize}

Now, we are ready to state the decomposition theorem for even-hole-free graphs given by Da Silva and Vu{\v{s}}kovi{\'c}.

\begin{theorem}[Da Silva, Vu{\v{s}}kovi{\'c} \cite{DV13}]
  A connected even-hole-free graph is either basic or it has a $2$-join or a star cutset.
\end{theorem}

By this theorem, we already know that even-hole-free graphs with no star cutset always have a $2$-join. But we might prefer something a bit stronger for our purpose. A $2$-join is called \textit{extreme} if one of its block of decomposition is basic. The two following lemmas (which can be found in Sections 3 and 4 in \cite{TV12}) say that: our blocks of decomposition with respect to a $2$-join remain in the class and our class is fully decomposable by extreme $2$-joins. This is convenient for an inductive proof.

\begin{lemma}[Trotignon, Vu{\v{s}}kovi{\'c}  \cite{TV12}] \label{remain} 
Let $G$ be a connected even-hole-free graph with no star cutset and $(X_1,X_2)$ is a $2$-join of $G$. Let $G_1$ be a block of decomposition with respect to this $2$-join. Then $G_1$ is a connected even-hole-free graph with no star cutset.
\end{lemma}

\begin{lemma}[Trotignon, Vu{\v{s}}kovi{\'c} \cite{TV12}] \label{decomposition} 
A connected even-hole-free graph with no star cutset is either basic or it has an extreme $2$-join.
\end{lemma}

By Lemmas \ref{remain} and \ref{decomposition}, we know that even-hole-free graphs with no star cutset can be fully decomposed into basic graphs using only extreme $2$-joins. However, we need a little more condition to avoid confliction between these $2$-joins, that is, every $2$-join we use is \emph{non-crossing}, meaning that every marker path in the process always lies entirely in one side of every following $2$-joins (the edges between $X_1$ and $X_2$ do not belong to any marker path). Now we define the \emph{$2$-join decomposition tree} for this purpose. Note that this definition we give here is not only for even-hole-free graphs with no star cutset, but also works in a more general sense. It is well defined for any graph class with its own basic graphs. Let $\mathbb{D}$ be a class of graphs and $\mathbb{B}\subseteq \mathbb{D}$ be the set of basic graphs in $\mathbb{D}$. Given a graph $G\in\mathbb{D}$, a tree $\mathbb{T}_G$ is a \emph{$2$-join decomposition tree} for $G$ if:
\begin{itemize}
	\item Each node of $\mathbb{T}_G$ is a pair $(H,S)$, where $H$ is a graph in $\mathbb{D}$ and $S$ is a set of disjoint flat paths of $H$.
	\item The root of $\mathbb{T}_G$ is $(G,\emptyset)$.
	\item Each non-leaf node of $\mathbb{T}_G$ is $(G',S')$, where $G'$ has a $2$-join $(X_1,X_2)$ such that the edges between $X_1$ and $X_2$ do not belong to any flat path in $S'$. Let $S_1,S_2\subseteq S'$ be the set of the flat paths of $S'$ in $G'[X_1]$, $G'[X_2]$, respectively (note that $S'=S_1\cup S_2$). Let $G_1$, $G_2$ be two blocks of decomposition of $G'$ with respect to this $2$-join with marker paths $P_2$, $P_1$, respectively. The node $(G',S')$ has two children, which are $(G_1,S_1\cup \{P_2\})$ and $(G_2,S_2\cup \{P_1\})$.
	\item Each leaf node of $\mathbb{T}_G$ is $(G',S')$, where $G'\in\mathbb{B}$.
\end{itemize}   

Note that by this definition, each set $S'$ in some node $(G',S')$ of $\mathbb{T}_G$ is properly defined in top-down order (from root to leaves). A $2$-join decomposition tree is called \emph{extreme} if each non-leaf node of it has a child which is a leaf node. 

\begin{lemma}[Trotignon, Vu{\v{s}}kovi{\'c}  \cite{TV12}]
\label{non-crossing}
Every connected even-hole-free graphs with no star cutset has an extreme $2$-join decomposition tree.
\end{lemma}

\begin{observation} \label{ob:1}
Every block of decomposition with respect to a $2$-join of a connected even-hole-free graph with no star cutset which is basic is either a long pyramid or an extended nontrivial basic graph.
\end{observation}

Let us review the definition of rank-width, which was first introduced in \cite{OS2006}. For a matrix
$M=\{ m_{ij}: i\in R,\, j\in C\}$ over a field $F$, let $\rk(M)$ denote its linear rank. If $X\subseteq R$, $Y\subseteq C$, then let $M[X,Y]$ be the submatrix $\{ m_{ij}: i\in X,\, j\in Y\}$ of $M$. We assume that adjacency matrices of graphs are matrices over $GF(2)$.

Let $G$ be a graph and $A$, $B$ be disjoint subsets of $V(G)$. Let $M$ be the adjacency matrix of $G$ over $GF(2)$. We define the \textit{rank} of $(A,B)$, denoted by $\rk_G(A,B)$, as $\rk(M[A,B])$. The \textit{cut-rank} of a subset $A\subseteq V(G)$, denoted by $\cutrk_G(A)$, is defined by $$\cutrk_G(A)=\rk_G(A,V(G)\setminus A).$$

A \textit{subcubic tree} is a tree such that the degree of every vertex is either one or three. We call $(T,L)$ a \textit{rank-decomposition} of $G$ if $T$ is a subcubic tree and $L$ is a bijection from $V(G)$ to the set of leaves of $T$. For an edge $e$ of $T$, the two connected components of $T\setminus e$ correspond to a partition $(A_e,V(G)\setminus A_e)$ of $V(G)$. The \textit{width} of $e$ of the rank-decomposition $(T,L)$ is $\cutrk_G(A_e)$. The \textit{width} of $(T,L)$ is the maximum width over all edges of $T$. The \textit{rank-width} of $G$, denoted by $\rwd(G)$, is the minimum width over all rank-decompositions of $G$ (If $|V(G)|\leq 1$, we define $\rwd(G)=0$). 

\begin{observation} \label{ob:2}
The rank-width of a clique is at most $1$ and the rank-width of a hole is at most $2$. 
\end{observation}

\section{$\chi$-bounding function} \label{S:3}

\subsection{Special graphs}

Recall that the bound of chromatic number for even-hole-free graphs ($\chi(G)\leq 2\omega(G)-1$) is obtained by showing that there is a vertex whose neighborhood is a union of two cliques \cite{ACHRS2008}. We would like to do the same things for our class. However, since our class is not closed under vertex-deletion, instead of showing that there exists a vertex whose neighborhood is ``simple'', we have to show that there is an elimination order such that the neighborhood of each vertex is ``simple'' in the remaining graph. To achieve that goal, we introduce \emph{special} graphs. In fact, this is just a way of labeling vertices for the sake of an inductive proof.

A graph $G$ is \textit{special} if it is associated with a pair $(C_G, F_G)$ such that:
\begin{itemize}
	\item $C_G\subseteq V(G)$, $F_G\subseteq V(G)$ and $C_G\cap F_G=\emptyset$.
	\item Every vertex in $F_G$ has degree $2$.
	\item Every vertex in $C_G$ has at least one neighbor in $F_G$. 
\end{itemize}
Note that any graph can be seen as a special graph with $C_G=F_G=\emptyset$.  

Suppose that $G$ has some split $(X_1, X_2, A_1, B_1, A_2, B_2)$ of a $2$-join.  Due to this new notion of special graph, we want to specify the pairs $(C_{G_1}, F_{G_1})$ and $(C_{G_2}, F_{G_2})$ for the blocks of decomposition $G_1$, $G_2$ of $G$ with respect to this $2$-join to ensure that the two blocks we obtained are also special. Let $C_i=C_G\cap X_i$, $F_i=F_G\cap X_i$ ($i=1,2$), we choose the pair $(C_{G_1}, F_{G_1})$ as follows: 
\begin{itemize}
	\item If $|A_1|=1$, the only vertex in $A_1$ is in $C_G$ and $A_2\cap F_G\neq \emptyset$, then set $C_a=\emptyset$, $F_a=\{a_2\}$. Otherwise set $C_a=\{a_2\}$, $F_a=\emptyset$.
	\item If $|B_1|=1$, the only vertex in $B_1$ is in $C_G$ and $B_2\cap F_G\neq \emptyset$, then set $C_b=\emptyset$, $F_b=\{b_2\}$. Otherwise set $C_b=\{b_2\}$, $F_b=\emptyset$.
	\item Finally, set $C_{G_1}=C_1\cup C_a\cup C_b$, $F_{G_1}=F_1\cup F_a\cup F_b\cup V(P_2^*)$.
\end{itemize}

The pair $(C_{G_2}, F_{G_2})$ for block $G_2$ is chosen similarly.

\begin{lemma} \label{remains}
Let $G$ be a special connected even-hole-free graph with no star cutset associated with $(C_G, F_G)$ and $(X_1, X_2, A_1, B_1, A_2, B_2)$ be a split of a $2$-join of $G$. Let $G_1$ be a block of decomposition with respect to this $2$-join. Then $G_1$ is a special graph associated with $(C_{G_1}, F_{G_1})$. 
\end{lemma}

\begin{proof}
Remark that since $G$ is $4$-hole-free, one of $A_1$ and $A_2$ must be a clique (similar for $B_1$ and $B_2$). Now we prove that if one of $A_1$ and $A_2$ intersects $F_G$, then the other set is of size $1$. Suppose that $A_1\cap F_G\neq \emptyset$, we will prove that $|A_2|=1$. Indeed, since $f\in A_1\cap F_G$ has degree $2$, $|A_2|\leq 2$. If $|A_2|=2$ then $f$ is the only vertex in $A_1$ (otherwise, $A_2$ must be a clique and $N(f)$ is a clique cutset separating $f$ from the rest of $G$, a contradiction to the fact that $G$ has no star cutset). Therefore, $f$ has no neighbor in $X_1$, so there is no path from $A_1$ to $B_1$ in $G[X_1]$, a contradiction to the definition of a $2$-join. This proves that $|A_2|=1$. Now, $G_1$ is a special graph associated with $(C_{G_1}, F_{G_1})$ because:

\begin{enumerate}
	\item Every vertex $f$ in $F_{G_1}$ has degree $2$.
	
If $f\in F_1\setminus (A_1\cup B_1)$, then degree of $f$ remains the same in $G$ and $G_1$. If $f\in F_1\cap (A_1\cup B_1)$, say $f\in F_1\cap A_1$, from the above remark, $|A_2|=1$, therefore the degree of $f$ remains the same in $G$ and $G_1$. If $f\in F_a\cup F_b$ then $|A_1|=1$ by the way we choose $F_{G_1}$, so $f$ has degree $2$ in $G_1$. If $f\in P_2^*$, then it is an interior vertex of a flat path, therefore it has degree $2$. 
	
	\item Every vertex $c$ in $C_{G_1}$ has at least a neighbor in $F_{G_1}$.
	
If $c\in C_1$ and its neighbor in $F_G$ is in $X_1$, then $c$ has a neighbor in $F_1$. If $c\in C_1$ and its neighbor in $F_G$ is in $A_2\cup B_2$, say $A_2$, then its neighbor in $F_{G_1}$ is $a_2$. If $c\in C_a\cup C_b$, then its neighbor in $F_{G_1}$ is one of the two ends of $P_2^*$.
\end{enumerate}
\end{proof}

\subsection{Elimination order} 
Let $G$ be a special graph associated with $(C_G, F_G)$. A vertex $v\in V(G)$ is \textit{almost simplicial} if its neighborhood induces a clique or a union of a clique $K$ and a vertex $u$ such that $u\notin C_G$ ($u$ can have neighbor in $K$). An elimination order $v_1$,\ldots, $v_k$ of vertices of $G\setminus F_G$ is \textit{nice} if for every $1\leq i\leq k$, $v_i$ is \textit{almost simplicial} in $G\setminus (F_G\cup \{v_1,\ldots,v_{i-1}\})$. The next lemma is the core of this section.  
\begin{lemma} \label{L1}
Let $G$ be a special connected even-hole-free graph with no star cutset associated with $(C_G,F_G)$. Then $G\setminus F_G$ admits a nice elimination order. 
\end{lemma} 
By setting $C_G=F_G=\emptyset$, we have the following corollary of Lemma \ref{L1}:
\begin{corollary}
Let $G$ be a connected even-hole-free graph with no star cutset. Then $G$ admits a nice elimination order.
\end{corollary}

Theorem \ref{T1} follows immediately from the above corollary since we can greedily color $G$ in the reverse order of that nice elimination order using at most $\omega(G)+1$ colors. Therefore, the rest of this section is devoted to the proof of Lemma \ref{L1}.  

\begin{lemma} \label{L2}
Let $G$ be a special basic even-hole-free graph with no star cutset associated with $(C_G, F_G)$ and $G$ is neither a clique nor a hole. Let $P$ be a flat path of length at least $2$ in $G$. We denote by $u_1$, $u_2$ the two ends of $P$. 
\begin{itemize}
	\item If $N(u_1)\setminus V(P)$ is a clique, set $K_1=N(u_1)\setminus V(P)$, otherwise set $K_1=\emptyset$. 
	\item If $N(u_2)\setminus V(P)$ is a clique, set $K_2=N(u_2)\setminus V(P)$, otherwise set $K_2=\emptyset$. 
\end{itemize}
Let $Q_P=(K_1\cup K_2 \cup V(P))\setminus F_G$. Then $G\setminus F_G$ admits a nice elimination order $v_1$,\ldots, $v_k$, where $Q_P$ is in the end of this order $($i.e. $Q_P=\{v_{k-|Q_P|+1},\ldots,v_k\})$.

\end{lemma}

\begin{proof}
We prove the lemma when $G$ is a long pyramid or an extended nontrivial basic graph. In fact, since the proof for a long pyramid can be treated almost similarly, we only show here the proof in the case where $G$ is an extended nontrivial basic graph. Suppose that $V(G)=V(H)\cup \{x,y\}$, where $H$ is the line graph of a tree. We may assume the followings:
\begin{enumerate}
	\item \label{as1} $P$ is a maximal flat path in $G$ (two ends of $P$ are of degree $\geq 3$).  
		
	If the lemma is true when $P$ is a maximal flat path then it is also true for all subpaths of $P$, because $Q_P$ admits a perfect elimination order (an order of vertices in which the neighborhood of a vertex induces a clique at the time it is eliminated) where a fixed subpath of $P$ is in the end of this order.	
	
	\item \label{as2} All the vertices in $G\setminus (F_G\cup Q_P\cup\{x,y\})$ are not in $C_G$.
	
	Observe that the neighborhood of every vertex $v$ in $G$, except $x$ and $y$, induces a union of two cliques. Therefore, if $v\in C_G$, it must have a neighbor of degree $2$ in $F_G$, then its neighborhood in $G\setminus F_G$ actually induces a clique and it can be eliminated at the beginning of our order.
	
	\item \label{as3} Every vertex in $G\setminus (F_G\cup Q_P)$ has at most one neighbor in $C_G$.		
		
	Indeed, by the assumption \ref{as2}, $C_G\subseteq \{x,y\}\cup Q_P$. If a vertex $v\in G\setminus (F_G\cup Q_P)$ has two neighbors in $C_G$, then it must have a neighbor $u\in (C_G\cap Q_P)\setminus\{x,y\}$. By the definition of $C_G$, $v$ must be a vertex in $F_G$ since it is the only neighbor of degree $2$ of $u$, a contradiction to the choice of $v$.	
	
\end{enumerate}

Let us first forget about the flat path $P$ and the restriction of putting all the vertices of $Q_P$ in the end of the order. We will show how to obtain a nice elimination order for $G\setminus F_G$ in this case. We choose an arbitrary extended clique $K_R$ in $H$ and call it the \emph{root} clique. For each other extended clique $K$ in $H$, there exists a vertex $v\in K$ whose removal separates the root clique from $K\setminus v$ in $H$, we call it \emph{B-vertex}. We call a node \emph{E-vertex} if it is adjacent to $x$ or $y$. Note that in each extended clique $K$, we have exactly one B-vertex and at most one E-vertex (by Lemma \ref{pr4}). For the root clique $K_R$, we also add a new vertex $r$ adjacent to all the vertices of $K_R$, and let it be the B-vertex for $K_R$. Now, if we remove every edge in every extended clique, except the edges incident to its B-vertex, we obtain a tree $T_H$ rooted at $r$. Note that $V(T_H)=V(H)\cup \{r\}$. We specify the nice elimination order for $G$ where all the vertices in $V(K_R)$ are removed last (we do not care about the order of eliminating $r$, this vertex is just to define an order for $V(G)$ more conveniently). Let $O_T$ be an order of visiting $V(T_H)\setminus (V(K_R)\cup \{r\})$ satisfying:
\begin{itemize}
	\item A node $u$ in $T_H$ is visited after all the children of $u$.
	\item If $u$ is a B-vertex of some extended clique $K$, the children of $u$ must be visited in an order where the E-vertex in $K$ (if it exists) is visited last. 
\end{itemize}

Let us introduce some notions with respect to orders first. Let $O_1=v_1,\ldots,v_k$ and $O_2=u_1,\ldots,u_t$ be two orders of two distinct sets of vertices. We denote by \emph{$O_1\oplus O_2$} the order $v_1,\ldots,v_k,u_1,\ldots,u_t$. If $S$ is a subset of vertices of some order $O_1$, we denote by \emph{$O_1\setminus S$} the order obtained from $O_1$ by removing $S$. Let $u$ be a vertex, we also denote by $u$ the order of one element $u$. 

Let $O_{K_R}$ be an arbitrary elimination order for the vertices in $K_R$. Now the elimination order for $G\setminus F_G$ is $O=(O_T\setminus F_G)\oplus x\oplus y\oplus O_{K_R}$. We prove that this elimination order is nice. Indeed, let $u$ be a vertex of order $O_T$, $u\notin F_G$. If $u$ is not an E-vertex, then its neighborhood at the time it is eliminated is either a subclique of some extended clique (if its parent is a B-vertex) or a single vertex which is its parent. If $u$ is an E-vertex, since it is eliminated after all its siblings (the nodes share the same parent), its neighborhood consists of only two vertices: its parent and $x$ (or $y$). And by assumption \ref{as3}, at most one of these two vertices is in $C_G$, so $u$ is almost simplicial. Now when $x$ is removed in this order, it has at most two neighbors: one is $y$ and one is possibly a vertex in $K_R$, also not both of them are in $C_G$, so $x$ is almost simplicial. Vertex $y$ has at most one neighbor at the time it is eliminated. And finally, $K_R$ is a clique so any eliminating order for $K_R$ at this point is nice.  

Now we have to consider the flat path $P$, and put all the vertices of $Q_P$ in the end of the elimination order. There are two cases: 
\begin{itemize}

	\item $P$ is a flat path not containing $x$ and $y$. 
	
	In this case, both $K_1$ and $K_2$ are non-empty. The graph obtained from $G$ by removing $V(P)\cup \{x,y\}$ contains two connected components $H_1$, $H_2$, where $H_i$ ($i=1,2$) is the line graph of a tree. By considering $K_i$ as the root clique of $H_i$, from the above argument, we obtain two elimination orders $O_1$, $O_2$ for $H_1\setminus K_1$ and $H_2\setminus K_2$. Now, all the vertices not yet eliminated in $G$ are in $\{x,y\}\cup Q_P$. We claim that at least one of $x$, $y$ has at most two neighbors in the remaining graph. Indeed, otherwise $x$ and $y$ both have at least three neighbors, implying that they both have neighbors in $K_1$ and $K_2$, which contradicts Lemma \ref{pr4}. Suppose $x$ has at most two neighbors, in this case we eliminate $x$ first, then $y$. Note that $x$ and $y$ are both almost simplicial in this elimination order, since they have at most two neighbors (not both of them in $C_G$ according to assumption \ref{as3}) at the time they were eliminated. Finally, choose for $Q_P$ a perfect elimination order $O_Q$. Now the nice elimination order for $G\setminus F_G$ is $O=((O_1\oplus O_2)\setminus F_G)\oplus x\oplus y\oplus O_Q$.   
	\item $x$ or $y$ is an end of $P$.
	
	W.l.o.g, suppose $x$ is an end of $P$, say $x=u_1$. In this case $K_1=\emptyset$, $K_2\neq \emptyset$. The graph $H'$ obtained by removing $V(P)\cup \{y\}$ from $G$ is the line graph of a tree. We consider $K_2$ as the root clique of this graph. From above argument, we have a nice elimination order $O_{H'}$ for $V(H')\setminus K_2$. Now all the vertices left in $G$ are in $\{y\}\cup Q_P$. Observe that $y$ has at most two neighbors in the remaining graph ($x$ and possibly a vertex in $K_2$), therefore $y$ is almost simplicial and can be eliminated. Note that $x$ has no neighbor in $K_2$, since if $u\in K_2$ is adjacent to $x$, then $\{u\}\cup N(u)$ is a star cutset in $G$ separating $P^*$ from the rest of the graph ($P^*$ is non-empty since the length of $P$ is at least two). Finally, choose for $Q_P$ a perfect elimination order $O_Q$. Now the nice elimination order for $G\setminus F_G$ is $O=(O_{H'}\setminus F_G)\oplus y\oplus O_Q$.      	
\end{itemize}
\end{proof}

\begin{proof}[Proof of Lemma \ref{L1}]
By Lemma \ref{non-crossing}, $G$ has an extreme 2-join decomposition tree $\mathbb{T}_G$. Now, for every node $(G',S')$ of $\mathbb{T}_G$, $G'$ is a special connected even-hole-free graph with no star cutset associated with $(C_{G'},F_{G'})$ (by Lemmas \ref{remain} and \ref{remains}). Now we prove that for every node $(G',S')$ of $\mathbb{T}_G$, $G'$ satisfies Lemma \ref{L1}. This implies the correctness of Lemma \ref{L1} since the root of $\mathbb{T}_G$ corresponds to $G$.

First, we show that for every leaf node $(G',S')$ of $\mathbb{T}_G$, $G'$ satisfies Lemma \ref{L1}. If $G'$ is a clique then any elimination order of $G'\setminus F_{G'}$ is nice. If $G'$ is a hole, there exists a vertex $v$ such that one of its neighbors is not in $C_{G'}$, then $v$ can be eliminated first. The vertices in the remaining graph induce a subgraph of a path, therefore $G'\setminus v$ admits a nice elimination order. If $G'$ is a long pyramid or an extended nontrivial basic graph, we have a nice elimination order for $G'$ by Lemma \ref{L2}.

Now, let us prove that Lemma \ref{L1} holds for $G'$, where $(G',S')$ is a non-leaf node of $\mathbb{T}_G$. Since $\mathbb{T}_G$ is extreme, $G'$ admits an extreme $2$-join with the split $(X_1, X_2, A_1, B_1, A_2, B_2)$ and let $G_1$, $G_2$ be the blocks of decomposition of $G'$ with respect to this $2$-join. We may assume that $G_1$ is basic and $G_2$ satisfies Lemma \ref{L1} by induction. Note that $V(G')=(V(G_1)\setminus V(P_2))\cup (V(G_2)\setminus V(P_1))$. Now we try to specify a nice elimination order for $G'$ by combining the orders for $G_1$ and $G_2$. Since $G_1$ is basic, apply Lemma \ref{L2} for $G_1$ with $P=P_2$, we obtain the nice elimination order $O_1$ for $G_1\setminus (F_{G_1}\cup Q_P)$. Remark that all the vertices in $O_1$ are in $V(G')$ since we have not eliminated $Q_P$. By induction hypothesis, we obtain also a nice elimination order $O_2$ for $G_2\setminus F_{G_2}$. We create an order $O_2'$ from $O_2$ for $V(G')$ as follows ($a_1$, $b_1$ are two ends of the marker path $P_1$):
\begin{itemize}
	\item If $a_1\in C_{G_2}$ and $A_1$ is a clique, $O_2'$ is obtained from $O_2$ by substituting $a_1$ in $O_2$ by all the vertices in $A_1$ (in any order), otherwise set $O_2'=O_2\setminus\{a_1\}$.
	\item If $b_1\in C_{G_2}$ and $B_1$ is a clique, $O_2'$ is obtained from itself by substituting $b_1$ in $O_2'$ by all the vertices in $B_1$ (in any order), otherwise set $O_2'=O_2'\setminus\{b_1\}$.
\end{itemize} 
We claim that $O=O_1\oplus O_2'$ is a nice elimination order for $G'\setminus F_{G'}$. Let $N'_{G'}(u)$ ($N'_{G_1}(u)$, $N'_{G_2}(u)$) be the set of neighbors of $u$ in the remaining graph when it is removed with respect to order $O$ ($O_1$, $O_2$, respectively). 
\begin{itemize}
	\item If $u$ is a vertex in $O_1$. 
	\begin{itemize}
		\item If $u\notin A_1$ and $B_1$, then $N_{G'}(u)=N_{G_1}(u)$, because $u$ is almost simplicial in $G_1$ then it is also almost simplicial in $G'$ at the time it was eliminated.
		\item If $u\in A_1$ or $B_1$, w.l.o.g, suppose that $u\in A_1$, then $A_1$ is not a clique, because we do not eliminate $Q_P$ in $O_1$. Since one of $A_1$, $A_2$ must be a clique to avoid $4$-hole, $A_2$ is a clique. If $a_2\in F_{G_1}$, then $|A_1|=1$ and $A_1$ is a clique of size $1$, a contradiction. Then $a_2\in C_{G_1}$. Because $a_2$ was not eliminated at the time we remove $u$ in order $O_1$, $a_2\in N'_{G_1}(u)$. We can obtain $N'_{G'}(u)$ from $N'_{G_1}(u)$ by substituting $a_2$ by $A_2$, therefore $u$ remains almost simplicial in $G'$. 
	\end{itemize}
	\item If $u$ is a vertex in $O_2'$.
	\begin{itemize}
		\item If $u\in X_2\setminus (A_2\cup B_2)$, then $N_{G'}(u)=N_{G_2}(u)$, because $u$ is almost simplicial in $G_2$ then it is also almost simplicial in $G'$ at the time it was eliminated.
		\item If $u\in A_2$ or $B_2$, w.l.o.g, suppose $u\in A_2$. We may assume that $A_1$ is a clique, since otherwise it was eliminated in $O_1$ before $u$, implying $N'_{G'}(u)\subseteq N'_{G_2}(u)$ and $u$ is almost simplicial in $G'$.
		\begin{itemize}
			\item Suppose $a_1\in C_{G_2}$. If $u$ is eliminated after $a_1$, then $N'_{G'}(u)=N'_{G_2}(u)$ and $u$ is almost simplicial in $G'$. If $u$ is eliminated before $a_1$, we can obtain $N'_{G'}(u)$ from $N'_{G_2}(u)$ by substituting $a_1$ by $A_1$, therefore $u$ remains almost simplicial.
			\item Suppose $a_1\in F_{G_2}$. Since $A_1$ is a clique and it contains a vertex $v\in F_{G'}$, $|A_1|\leq 2$. If $|A_1|=2$, then $N_{G'}(v)$ is a clique cutset of size $2$ (star cutset) separating $v$ from the rest of $G'$, a contradiction. Thus $A_1=\{v\}$ and $N'_{G'}(u)=N'_{G_2}(u)$ (since $v$ is the only vertex in $A_1$ and $v\notin G'\setminus F_{G'}$) and $u$ is almost simplicial in $G'$.  
		\end{itemize}
		\item If $u\in A_1$ or $B_1$, w.l.o.g, suppose $u\in A_1$, then $A_1$ is a clique, since otherwise it was removed in $O_1$. We can obtain $N'_{G'}(u)$ from $N'_{G_2}(a_1)$ by creating a clique $K$, which is a subclique of $A_1$ ($K$ is actually the set of vertices of $A_1$ going after $u$ in $O_2'$), and make it complete to $N'_{G_2}(a_2)$. Therefore, $u$ remains almost simplicial in $G'$.
	\end{itemize}
\end{itemize}
\end{proof}

\subsection{The bound is tight}

\begin{figure}[h]
\centering
\includegraphics[width=11cm]{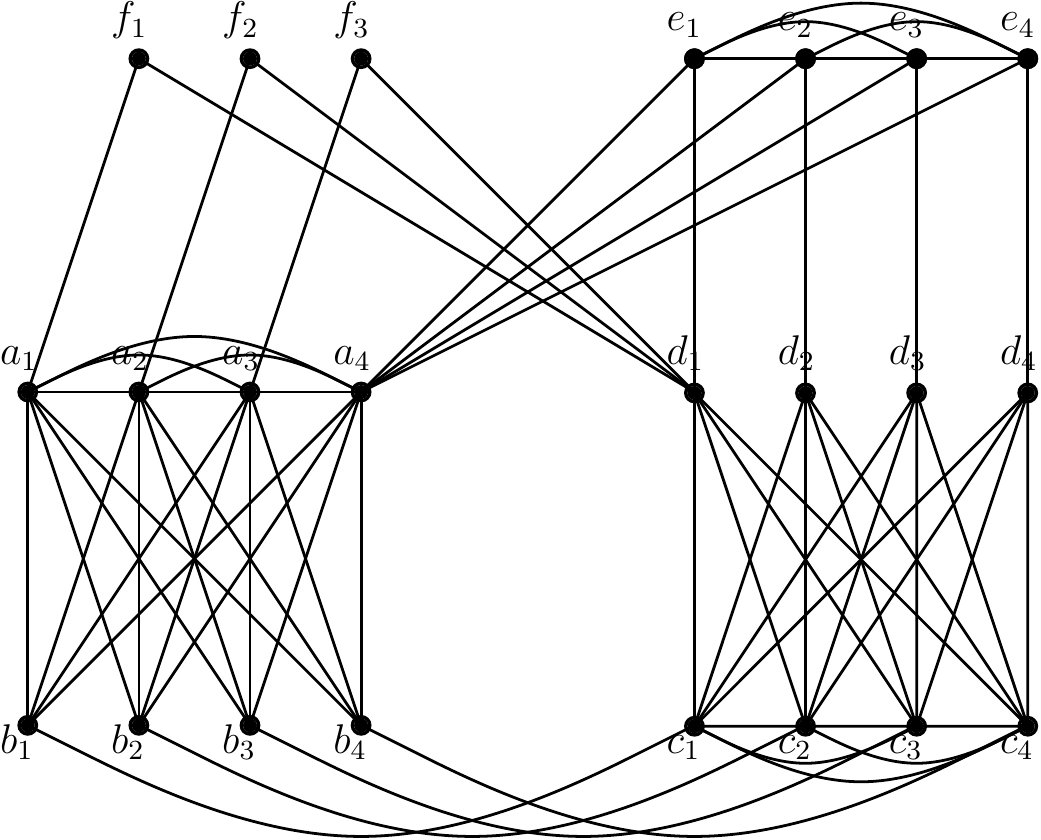}
\caption{Graph $G_5$ with $\omega(G_5)=5$ and $\chi(G_5)=6$}
\label{F:1}
\end{figure}

Now we show how to construct for any $k\geq 3$ an even-hole-free graph $G_k$ with no star cutset such that $\omega(G_k)=k$ and $\chi(G_k)=k+1$. The set of vertices of $G_k$: $V(G_k)=A\cup B\cup C\cup D\cup E\cup F$, where $A$, $C$, $E$ are cliques of size $(k-1)$; $B$, $D$ are independent sets of size $(k-1)$ and  $F$ is an independent set of size $(k-2)$. The vertices in each set are labeled by the lowercase of the name of that set plus an index, for example $A=\{a_1,\ldots,a_{k-1}\}$. The edges of $G_k$ as follows:
\begin{itemize}
	\item $A$ is complete to $B$, $C$ is complete to $D$.
	\item $b_i$ is adjacent to $c_i$, $d_i$ is adjacent to $e_i$  ($i=1,\ldots,k-1$).
	\item $a_{k-1}$ is complete to $E$.
	\item $d_1$ is complete to $F$.
	\item $a_i$ is adjacent to $f_i$ ($i=1,\ldots,k-2$).
\end{itemize} 
Figure \ref{F:1} is an example of $G_k$, where $k=5$. The fact that $G_k$ is an even-hole-free graph with no star cutset can be checked by hand.

\begin{lemma}
For every $k\geq 3$, $\omega(G_k)=k$ and $\chi(G_k)=k+1$.
\end{lemma}

\begin{proof}
It is clear that $\omega(G_k)=k$. We will show that $G_k$ is not $k$-colorable. Suppose we have a $k$-coloring of $G$. Because in that coloring, every clique of size $k$ must be colored by all $k$ different colors, then all the vertices in $B$ must receive the same color $1$. Therefore, the clique $C$ must be colored by $(k-1)$ left colors, and  all the vertices in $D$ must be colored by color $1$ also. Therefore, the $k$-clique $\{a_{k-1},e_1,\ldots,e_{k-1}\}$ is not colorable since all of them must have color different from $1$, a contradiction.
\end{proof}

\section{Rank-width} \label{S:4}

\subsection{Bounded rank-width}

Recall that the definition of rank-width and rank-decomposition are given in the last part of Section \ref{S:2}. Given a graph $G$ and some rank-decomposition $(T,L)$ of $G$, a subset $X$ of $V(G)$ is said to be \emph{separated} in $(T,L)$ if there exists an edge $e_X$ of $T$ corresponding to the partition $(X,V(G)\setminus X)$ of $V(G)$. Let $d$ be an integer, we say that graph $G$ has \emph{property $\mathbb{P}(d)$} if for every set $S$ of disjoint flat paths of length at least $3$ in $G$, there is a rank-decomposition $(T,L)$ of $G$ such that the width of $(T,L)$ is at most $d$ and every flat path $P\in S$ is separated in $(T,L)$. The next lemma shows the relation between $2$-join and rank-width.

\begin{lemma} \label{rank-2-join}
Let $\mathbb{D}$ be a class of graphs and $\mathbb{B}\subseteq\mathbb{D}$ be the set of its basic graphs such that every graph $G\in \mathbb{D}$ has a $2$-join decomposition tree. Furthermore, there exists an integer $d\geq 2$ such that every basic graph in $\mathbb{D}$ has property $\mathbb{P}(d)$. Then for every graph $G\in \mathbb{D}$, $\rwd(G)\leq d$. 
\end{lemma}

\begin{proof}
Let $G$ be a graph in $\mathbb{D}$ and $\mathbb{T}_G$ be its $2$-join decomposition tree. We prove that every node $(G',S')$ of $\mathbb{T}_G$ satisfies the following \emph{property $\mathbb{P'}(d)$}: there is a rank-decomposition $(T,L)$ of $G'$ such that the width of $(T,L)$ is at most $d$ and every flat path $P\in S'$ is separated in $(T,L)$. Note that property $\mathbb{P'}(d)$ is weaker than property $\mathbb{P}(d)$ since it is not required to be true for every choice of the set of disjoint flat paths, but only for a particular set $S'$ associated with $G'$ in $\mathbb{T}_G$. Proving this property for each node in $\mathbb{T}_G$ implies directly the lemma since if the root of $\mathbb{T}_G$ has property $\mathbb{P'}(d)$, then $\rwd(G)\leq d$. 

It is clear that every leaf node of $\mathbb{T}_G$ has property $\mathbb{P'}(d)$ since every basic graph has property $\mathbb{P}(d)$ by the assumption. Now we only have to prove that every non-leaf node $(G',S')$ of $\mathbb{T}_G$ has property $\mathbb{P'}(d)$ assuming that its two children $(G_1,S_1)$ and $(G_2,S_2)$ already have property $\mathbb{P'}(d)$. For $i\in\{1,2\}$, let $(T_i,L_i)$ be the rank-decomposition of $G_i$ satisfying property $\mathbb{P'}(d)$. We show how to construct the rank-decomposition $(T,L)$ of $G'$ satisfying this property. Recall that by the definition of a $2$-join decomposition tree, $G_1$ and $G_2$ are two blocks of decomposition with respect to some $2$-join $(X_1,X_2)$ of $G'$ together with some marker paths $P_2\in S_1$, $P_1\in S_2$, respectively. For $i\in\{1,2\}$, since $(G_i,S_i)$ satisfies property $\mathbb{P'}(d)$, $P_{3-i}$ is separated in $(T_i,L_i)$ by some edge $e_i=u_iv_i$ of $T_i$. Let $C_i$, $D_i$ be the two connected components (subtrees) of $T_i\setminus e_i$ (the tree obtained from $T_i$ by removing the edge $e_i$), where the leaves of $C_i$ correspond to $V(G_i)\setminus V(P_{3-i})$ and the leaves of $D_i$ correspond to $V(P_{3-i})$. W.l.o.g, we may assume that $u_i$ is in $C_i$ and $v_i$ is in $D_i$. The tree $T$ is then constructed from $T_1[V(C_1)\cup\{v_1\}]$ and $T_2[V(C_2)\cup\{v_2\}]$ by identifying $u_1$ with $v_2$ and $u_2$ with $v_1$. Note that $T$ is a subcubic tree and the leaves of $T$ now correspond to $V(G)$. The mapping $L$ is the union of the two mappings $L_1$ and $L_2$ restricted in $X_1$ and $X_2$, respectively. Now the node $(G',S')$ satisfies property $\mathbb{P}'_d$ since:
\begin{itemize}
	\item Every flat path $P\in S'$ is separated in $(T,L)$.
	
	It is true since for $i\in\{1,2\}$, every path $P\in S_i$ is separated in $(T_i,L_i)$.
	\item The width of $(T,L)$ is at most $d$.
	
	It is easy to see that the width of the identified edge $e=u_1v_1$ of $T$ is $2$, since it corresponds to the partition $(X_1,X_2)$ of $G'$. For other edge $e$ of $C_i$ (for $i=1$ or $2$), it corresponds to a cut of $G'$ separating a subset $Z$ of $X_i$ from $V(G')\setminus Z$, and we have $\cutrk_{G'}(Z)=\cutrk_{G_i}(Z)$ (since the rank of the corresponding matrix stays the same if we just add several copies of the columns corresponding to the two ends of the marker path $P_{3-i}$), which implies that $\cutrk_{G'}(Z)\leq d$.
\end{itemize}  
\end{proof}

Thanks to Lemma \ref{rank-2-join} and the existence of a $2$-join decomposition tree by Lemma \ref{non-crossing},  to prove that the rank-width of even-hole-free graphs with no star cutset is at most $3$, we are left to only prove that every basic even-hole-free graph with no star cutset has property $\mathbb{P}(3)$. Actually, by Observation \ref{ob:1}, we do not have to prove it for cliques and holes, since they never appear in the leaf nodes of any $2$-join decomposition tree of any graph in our class. Therefore, Theorem \ref{T2} is a consequence of Observation \ref{ob:2} and the following lemma:

\begin{lemma} \label{rank-basic}
Every basic even-hole-free graph with no star cutset, which is neither a clique nor a hole, has property $\mathbb{P}(3)$.
\end{lemma}

\begin{proof}
Let $G$ be a basic even-hole-free graph with no star cutset, which is different from a clique and a hole. Since $G$ is basic and $G$ is neither a clique nor a hole, $G$ must be an extended nontrivial basic graph or a long pyramid. Since the case where $G$ is a long pyramid can be followed easily from the case where it is an extended nontrivial basic graph. We omit the details for long pyramids here.

Let $G$ be an extended nontrivial basic graph, $V(G)=V(H)\cup \{x,y\}$, where $H$ is the line graph of a tree. Let $S$ be some set of flat paths of length at least $3$ in $G$. Now we show how to build the rank-decomposition of $G$ satisfying the lemma.

First, we construct the \textit{characteristic} tree $F_H$ for $H$. We choose an arbitrary extended clique in $H$ as a \textit{root} clique. Let $E$ be the set of flat paths obtained from $H$ by removing all the edges of every extended clique in $H$. Now, we define the father-child relation between two flat paths in $E$. A path $B$ is the \emph{father} of some path $B'$ if they have an endpoint in the same extended clique in $H$ and any vertex of $B$ is a cut-vertex in $H$ which separates $B'$ from the root clique. If $B$ is the father of $B'$ then we also say that $B'$ is a \emph{child} of $B$. Any path in $E$ which has only one endpoint in an extended clique is called \textit{leaf} path, otherwise it is called \textit{internal} path. Now, we consider each path $B$ in $E$ as a vertex $v_B$ in the characteristic tree $F_H$, and associate with each node $v_B$ a set $S_{v_B}=V(B)$. Each leaf path corresponds to a leaf in $F_H$ and each internal path corresponds to an internal node in $F_H$, which reserves the father-child relation (if a path $B$ is the father of some path $B'$ then $v_B$ is the father of $v_{B'}$ in $F_H$). We also add a root $r$ for $F_H$, and the children of $r$ are all the vertices $v_B$, where $B$ is a path with an endpoint in the root clique, let $S_r=\emptyset$. Now, we add two special vertices $x$, $y$ to attain the \textit{characteristic} tree $F_G$ for $G$. If $x$ (or $y$) is an endpoint of some flat path $P$ in $S$, then we set $S_v=S_v\cup\{x\}$ ($S_v=S_v\cup\{y\}$, respectively), where $v$ is the leaf in $F_H$ corresponding to the leaf path in $E$ which contains $P\setminus\{x\}$ ($P\setminus\{y\}$, respectively). Otherwise set $S_v=S_v\cup\{x\}$ ($S_v=S_v\cup\{y\}$), where $v$ is a leaf in $F_H$ corresponding to any path in $E$ having an endpoint adjacent to $x$ ($y$, respectively). Figures \ref{F:2} and \ref{F:3} are the example of an extended nontrivial basic graph $G$ and its characteristic tree $F_G$ (the bold edges are the edges of flat paths in $S$). Note that each node in $F_G$ corresponds to a subset $S_v$ of $V(G)$, they are all disjoint, each of them induces a flat path in $G$ and $V(G)=\cup_{v\in F_G}S_v$. 

\begin{figure}[h]
\centering
\includegraphics[width=7cm]{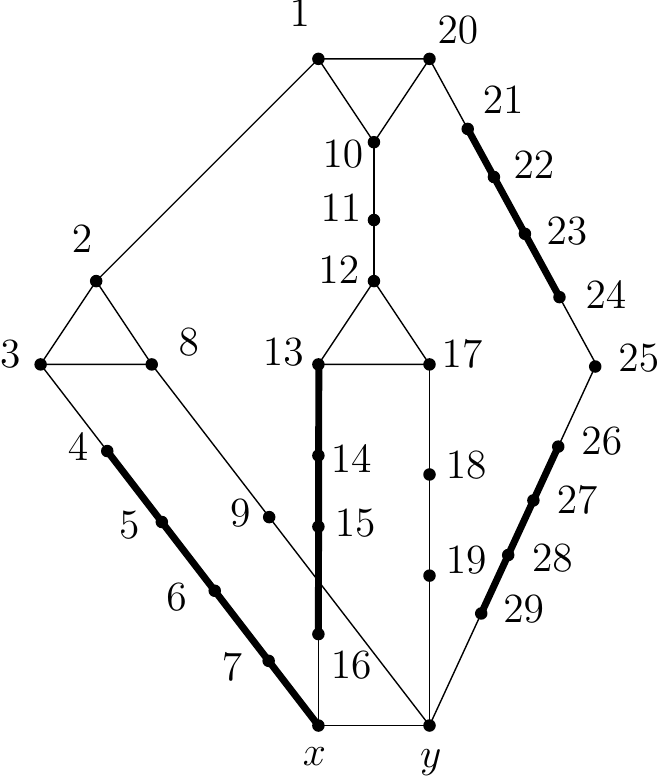}
\caption{An extended nontrivial basic graph $G$ with a set of flat paths.}
\label{F:2}
\end{figure}

\begin{figure}[h]
\centering
\includegraphics[width=11cm]{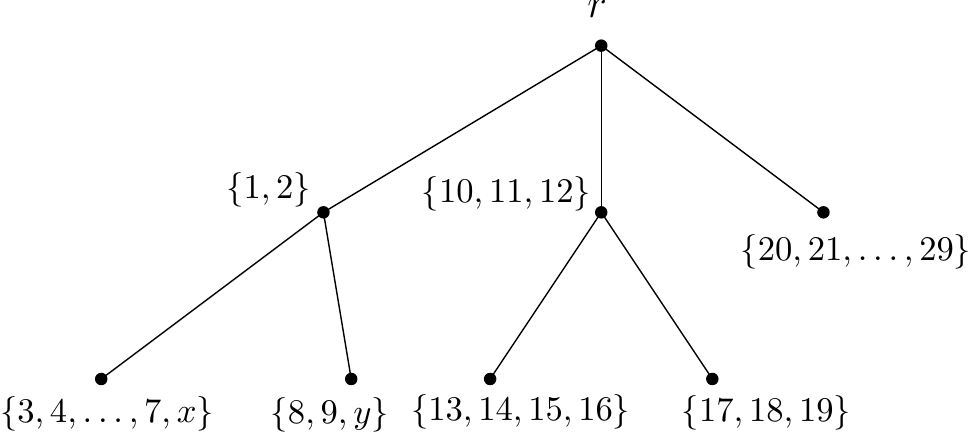}
\caption{The characteristic tree $F_G$ for graph $G$ in Figure \ref{F:2}.}
\label{F:3}
\end{figure}

Now, we show how to build the rank-decomposition of $G$ from its characteristic tree $F_G$. We first define a special rooted tree, called \emph{$k$-caterpillar} ($k\geq 1$) to achieve that goal. For $k\geq 1$, a graph $I$ is called \emph{$k$-caterpillar} if:
\begin{itemize}
	\item For $k=1$, $V(I)=\{a_1,l_1\}$, $E(I)=\{a_1l_1\}$ and $a_1$ is the root of $I$.
	\item For $k\geq 2$, $V(I)=\{a_1,\ldots,a_{k-1}\}\cup \{l_1,\ldots,l_{k}\}$, $E(I)=\{a_ia_{i+1}|1\leq i\leq k-2\}\cup \{a_il_i|1\leq i\leq k-1\}\cup \{a_{k-1}l_k\}$ and $a_1$ is the root of $I$. 
\end{itemize}

Notice that in the following discussion, for the sake of construction, the rank-decomposition $(T,L)$ we build for our graph is not exactly the same as in the definition of a rank-decomposition mentioned in Section \ref{S:2}, since we allow vertex of degree $2$ in tree $T$, but it does not change the definition of rank-width. A flat path in $G$ is called \textit{mixed} if it contains a flat path in $S$ but it is not a flat path in $S$. We start by constructing the rank-decomposition of a non-mixed flat path in $G$. For a non-mixed flat path $P=p_1\ldots p_k$, we create a $k$-caterpillar $T_P$ which has exactly $k$ leaves $l_1,\ldots,l_k$ as in the definition and a bijection $L_P$ maps each vertex in $P$ to a leaf of $T_P$ such that $L_P(p_i)=l_i$. Since a mixed path can always be presented as a union of vertex-disjoint non-mixed paths $P=\cup_{i=1}^kP_i$ (where one end of $P_i$ is adjacent to one end of $P_{i+1}$ for $1\leq i\leq k-1$), let $(T_i,L_i)$ be the rank-decomposition for each non-mixed path $P_i$ constructed as above, we can build the tree $T_P$ by creating a $k$-caterpillar $I$ which has exactly $k$ leaves $l_1,\ldots,l_k$ as in the definition and identify each root of $T_i$ with the leaf $l_i$ of $I$ for $1\leq i\leq k$. Also, let the mapping $L_P$ from $V(P)$ to the leaves of $T_P$ be the union of all the mappings $L_i$'s for $1\leq i\leq k$. Now, we build the rank-decomposition $(T_G,L_G)$ of $G$ from its characteristic tree $F_G$ by visiting each node in $F_G$ in an order where all the children of any internal node is visited before its father. For a vertex $v\in F_G$, denote by $C_v$ the union of all connected components of $F_G\setminus v$ that does not contain $r$. Let $F_G(v)=F_G[V(C_v)\cup\{v\}]$, $X_v=\cup_{u\in F_G(v)}S_u$. At each node $v$ of $F_G$, we build the rank-decomposition $(T_v,L_v)$ of the graph $G_v$ induced by the subset $X_v$ of $V(G)$ by induction:
\begin{enumerate}
	\item If $v$ is a leaf of $F_G$, build the rank-decomposition $(T_v,L_v)$ for the flat path corresponding to $v$ like above argument for mixed and non-mixed paths.
	\item If $v$ is an internal node of $F_G$ different from its root and $v_1,\ldots,v_k$ are its children. Let $(T,L)$ be the rank-decomposition of the flat path corresponding to $v$ (built by above argument for mixed and non-mixed paths) and $(T_i,L_i)$ ($i=1\ldots k$) be the rank-decomposition of $G[X_{v_i}]$. We build $T_v$ by constructing a $(k+1)$-caterpillar having exactly $(k+1)$ leaves $l_1,\ldots,l_{k+1}$ as in the definition and identify the root of $T$ with $l_1$, the root of $T_i$ with $l_{i+1}$ for $1\leq i\leq k$. Let the mapping $L_v$ from $X_{v_i}$ to the leaves of $T_v$ be the union of the mapping $L$ and all the mappings $L_i$'s for $1\leq i\leq k$.  
	\item If $v$ is the root of $F_G$ and $v_1,\ldots,v_k$ are its children. Let $(T_i,L_i)$ ($i=1\ldots k$) be the rank-decompositions of $G[X_{v_i}]$. We build $T_v$ by constructing a $k$-caterpillar having exactly $k$ leaves $l_1,\ldots,l_{k}$ as in the definition and identify the root of $T_i$ with $l_i$ for $1\leq i\leq k$. Let the mapping $L_v$ from $V(G)$ to the leaves of $T_v$ be the union of all the mappings $L_i$'s for $1\leq i\leq k$.  
\end{enumerate}   

\begin{figure}[h]
\centering
\includegraphics[width=12cm]{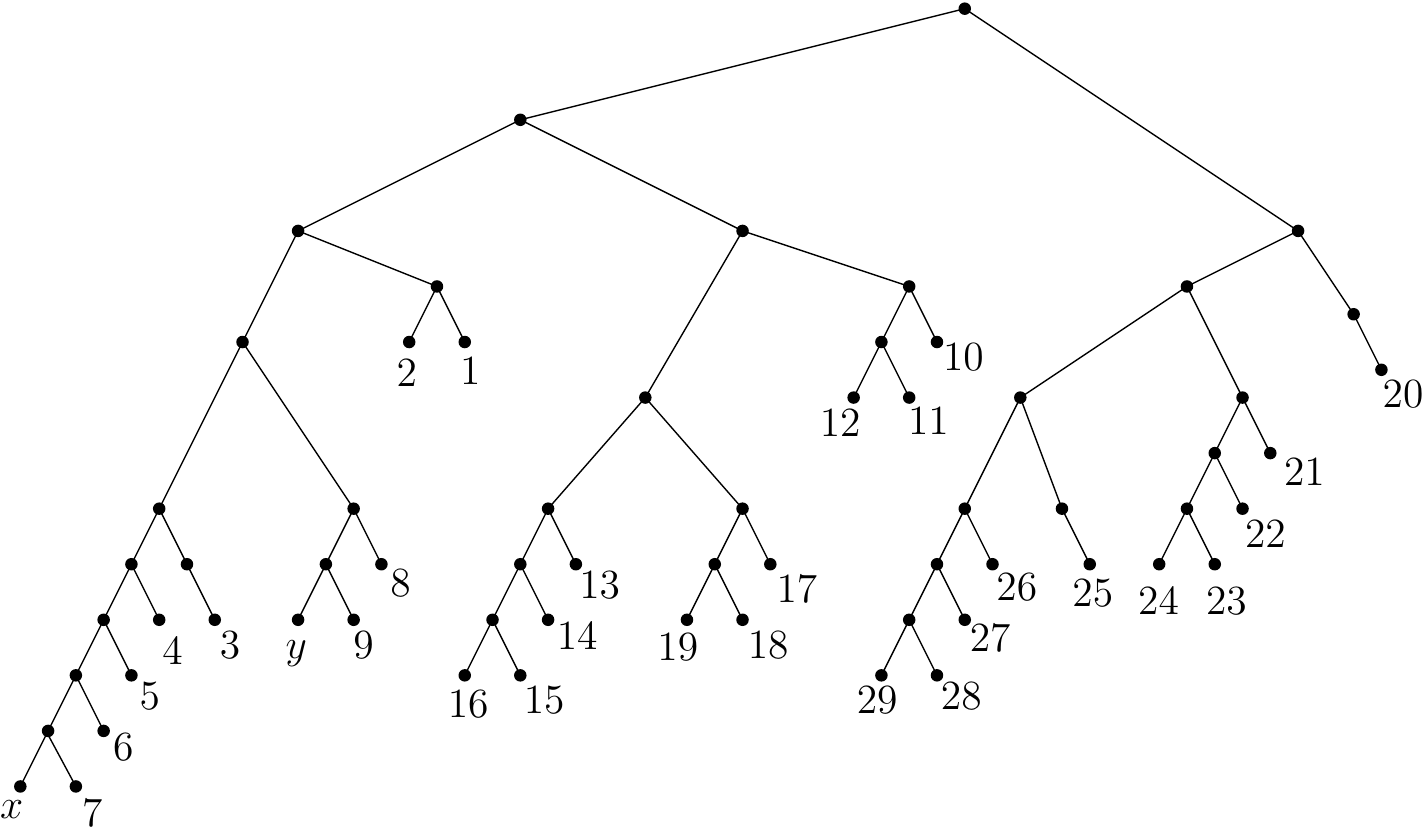}
\caption{The rank-decomposition for graph $G$ in Figure \ref{F:2}.}
\label{F:4}
\end{figure}

The rank-decomposition $(T_r,L_r)$ corresponding to the root $r$ of $F_G$ is the desired rank-decomposition $(T,L)$ for $G$ (see Figure \ref{F:4}). Now we prove that this rank-decomposition construction for the extended nontrivial basic graphs $G$ satisfies the lemma.

\begin{proposition} \label{prop1}
Let $(T,L)$ be the above constructed rank-decomposition for $G$. Then, every flat path $P$ in $S$ is separated in $(T,L)$.
\end{proposition}

\begin{proof}
It is trivially true, because $P$ is a non-mixed subpath of some flat path $B$ in $E$, so $V(P)$ is separated in the rank-decomposition of $B$. And each flat path $B$ of $G$ is also separated in the rank-decomposition of $G$ by our construction. So $V(P)$ is separated in $(T,L)$. 
\end{proof}

\begin{proposition} \label{prop2}
The above constructed rank-decomposition $(T,L)$ of $G$ has width at most $3$. 
\end{proposition}
 
\begin{proof}
We prove by the structure of the characteristic tree $F_G$ of $G$. For an internal node $v$ of $F_G$, let $v_1,\ldots,v_k$ be its children, in some sense, the decomposition tree $T_v$ for $X_v$ is obtained by ``glueing'' the decomposition tree for $G[S_v]$ and all the decomposition trees $T_i$ for $G[X_{v_i}]$ for $1\leq i\leq k$ along a cut-vertex. Therefore, we consider an edge $e$ of $T_v$ as an edge of $T$ as well. Our goal is to prove that the width of any edge $e$ with respect to the rank-decomposition $(T,L)$ of $G$ is at most 3. For the sake of induction, at each node $v$ of $F_G$, we prove that the width of any edge $e$ of $T_v$ is at most $3$ with respect to the rank-decomposition $(T,L)$ of $G$ (we mention $v$ here just to specify an edge in our tree $T$):
\begin{enumerate}
	\item If $v$ is a leaf in $F_G$. Every edge $e$ of $T_v$ corresponds to a partition of $V(G)$ into two parts where one of them is a subpath of the flat path corresponding to $v$, so the width of $e$ is at most $2$.
	\item If $v$ is an internal node in $F_G$ and $v_1,\ldots,v_k$ are its children. Let $(T_i,L_i)$ ($i=1\ldots k$) be the rank-decompositions of $G[X_{v_i}]$. Let $e$ be an edge of $T_v$. If $e$ is an edge of $T_i$ then the width of $e$ is at most $3$ by induction. Otherwise, $e$ corresponds to one of the following situations:
	\begin{itemize}
		\item $e$ corresponds to a partition $(V(P),V(G)\setminus V(P))$ of $V(G)$, where $P$ is a subpath of the flat path $G[S_v]$. In this case, the width of $e$ is clearly at most $2$.
		\item $e$ corresponds to a partition $(U,V(G)\setminus U)$ of $V(G)$, where $U$ is the union of several $X_{v_i}$'s. Let $K$ be the extended clique intersecting every $X_{v_i}$. In this case, there are only three types of neighborhood of vertices of $U$ in $G\setminus U$:
		\begin{itemize}
			\item $K\setminus U$, 
			\item $x$ if $x\notin U$, or $N(x)\setminus U$ if $x\in U$, and
			\item $y$ if $y\notin U$, or $N(y)\setminus U$ if $y\in U$.
		\end{itemize}	
		Therefore, the width of $e$ is at most $3$. 	 
	\end{itemize}
\end{enumerate}  
\end{proof} 
Lemma \ref{rank-basic} is true because of the Propositions \ref{prop1} and \ref{prop2}.
\end{proof}

\subsection{An even-hole-free graph with no clique cutset and unbounded rank-width}

It is clear that clique cutset is a particular type of star cutset. However, the class of even-hole-free graph with no clique cutset (a super class of even-hole-free graph with no star cutset) does not have bounded rank-width. Since clique-width and rank-width are equivalent, now we show how to construct for every $k\geq 4$, $k$ even an even-hole-free graph $G_k$ with no clique cutset and $\cwd(G_k)\geq k$. The set of vertices of $G_k$: $V(G_k)=\cup_{i=0}^{k}A_i$, where each $A_i=\{a_{i,0},\ldots,a_{i,k}\}$ is a clique of size $(k+1)$. We also have edges between two consecutive sets $A_i$, $A_{i+1}$ ($i=0,\ldots,k$, the indexes are taken modulo $(k+1)$). They are defined as follows: $a_{i,j}$ is adjacent to $a_{i+1,l}$ iff $j+l\leq k$.
\begin{lemma}
For every $k\geq 4$, $k$ even, $G_k$ is an even-hole-free graph with no clique cutset.
\end{lemma}  
\begin{proof}
By the construction, there is no hole in $G_k$ that contains two vertices in some set $A_i$ and every hole must contain at least a vertex in each set $A_i$. Therefore, every hole in $G_k$ has exactly one vertex from each set $A_i$, so its length is $(k+1)$ (an odd number). Hence, $G_k$ is even-hole-free. 

We see that every clique in $G_k$ is contained in the union of some two consecutive sets $A_i$, $A_{i+1}$. Hence, its removal does not disconnect $G_k$. Therefore, $G_k$ has no clique cutset.
\end{proof}

\begin{lemma}
For every $k\geq 4$, $k$ even, $\cwd(G_k)\geq k$.
\end{lemma}
\begin{proof}
The graph obtained from $G_k$ by deleting all the vertices in $A_0\cup_{i=1}^{k} \{a_{i,0}\}$ is isomorphic to the permutation graph $H_k$ introduced in \cite{GR2000}. And because it was already proved in that paper that $\cwd(H_k)\geq k$, and clique-width of $G_k$ is at least the clique-width of any of its induced subgraph then $\cwd(G_k)\geq k$. 
\end{proof}

Note that another example of an even-hole-free graph with no clique cuset and unbounded rank-width is also given in \cite{VTRMLA17}.

\bibliographystyle{abbrv}
\bibliography{color_even_hole}

\end{document}